\documentclass[final,onefignum,onetabnum]{siamonline190516}

\usepackage{lipsum}
\usepackage{amsfonts}
\usepackage{graphicx}
\usepackage{epstopdf}
\usepackage{algorithmic}
\usepackage{pdfpages}

\ifpdf
  \DeclareGraphicsExtensions{.eps,.pdf,.png,.jpg}
\else
  \DeclareGraphicsExtensions{.eps}
\fi

\usepackage{enumitem}
\setlist[enumerate]{leftmargin=.5in}
\setlist[itemize]{leftmargin=.5in}

\newsiamremark{remark}{Remark}
\newsiamremark{hypothesis}{Hypothesis}
\crefname{hypothesis}{Hypothesis}{Hypotheses}
\newsiamthm{claim}{Claim}

\headers{Leveraging Sharing Communities to Achieve \ldots}{F. W. Bentrem, M. A. Corsello, and J. J. Palm}

\title{Leveraging Sharing Communities to Achieve Federated Learning for Cybersecurity\thanks{Submitted to the editors March 5, 2021; accepted for publication March 25, 2021; revised April 15, 2021.
\funding{This work was funded by the Defense Advanced Research Projects Agency under contract no.~W911NF-18-C-0019 through subcontract under the University of Virginia.}}}

\author{Frank W. Bentrem\thanks{Commonwealth Computer Research, Inc., Charlottesville, VA 
  (\email{frank.bentrem@ccri.com}, \\ \email{michael.corsello@ccri.com}, \email{joshua.palm@ccri.com}, \url{https://ccri.com}).}
\and Michael A. Corsello\footnotemark[2]
\and Joshua J. Palm\footnotemark[2]}

\usepackage{amsopn}

\makeatletter
\newcommand*{\addFileDependency}[1]{
  \typeout{(#1)}
  \@addtofilelist{#1}
  \IfFileExists{#1}{}{\typeout{No file #1.}}
}
\makeatother

\newcommand*{\myexternaldocument}[1]{%
    \externaldocument{#1}%
    \addFileDependency{#1.tex}%
    \addFileDependency{#1.aux}%
}

\ifpdf
\hypersetup{
  pdftitle={Leveraging Sharing Communities to Achieve Federated Learning for Cybersecurity},
  pdfauthor={F.~W. Bentrem, M.~A. Corsello, and J.~J. Palm}
}
\fi

\myexternaldocument{SIAM2021_supplement}

\begin{document}

\maketitle

\begin{abstract}
  Automated cyber threat detection in computer networks is a major challenge in cybersecurity. The cyber domain has inherent challenges that make traditional machine learning techniques problematic, specifically the need to learn continually evolving attacks through global collaboration while maintaining data privacy, and the varying resources available to network owners. We present a scheme to mitigate these difficulties through an architectural approach using community model sharing with a streaming analytic pipeline. Our streaming approach trains models incrementally as each log record is processed, thereby adjusting to concept drift resulting from changing attacks. Further, we designed a community sharing approach which federates learning through merging models without the need to share sensitive cyber-log data. Finally, by standardizing data and Machine Learning processes in a modular way, we provide network security operators the ability to manage cyber threat events and model sensitivity through community member and analytic method weighting in ways that are best suited for their available resources and data.
\end{abstract}



\section{Introduction}
Computers have been subject to malicious attacks for many years. However, the sophistication and impacts of those attacks have grown dramatically. Successful, timely detection of such attacks is essential to enable defenders to prevent or respond to an attack to mitigate or minimize damage. The consequences from such attacks include substantial financial losses, leaked personal data, and the release of proprietary information or state secrets. The immense human and computational resources required to analyze network logs is prohibitive and often still misses novel attacks. As a result, there is a growing need to incorporate advanced techniques such as machine learning in the cyber defender’s toolbox. However, the cyber domain presents key challenges to traditional batch machine learning techniques \cite{apruzzese2018effectiveness, conti2018cyber}. 

\vspace{5pt}
\indent \textit{Problem \thesection.1}
The continuous evolution of threat vectors requires continual learning. 

\vspace{5pt}
\indent \textit{Problem \thesection.2}
The volume and sensitivity of network log data prevents long-term retention and sharing across organizations. 

\vspace{5pt}
\indent \textit{Problem \thesection.3}
Barriers-to-entry prevent small or more vulnerable organizations from adopting complex or expensive technology. 

\vspace{5pt}
\indent \textit{Problem \thesection.4}
Feedback from network security operators is important to identify novel attacks and mitigate as quickly as possible.
\vspace{5pt}

In this paper, we present solutions to these problems through an architecture that (1) is streaming, (2) is federated, and (3) manages feedback.
The paper is organized in the following way. Our targeted constraints are in
\cref{sec:constraints}, our system architecture is in \cref{sec:architecture}, the implemented machine learning algorithms
are in \cref{sec:algorithms}, and the summary follows in
\cref{sec:summary}.

\section{Targeted Constraints}
\label{sec:constraints}

Individual organizations that train models to recognize cyber threats are severely data limited, however, privacy concerns prevent them from directly sharing raw network data logs. Even anonymized data may collectively reveal sensitive information. We therefore propose a federated learning system \cite{konevcny2016federated} where trained standardized models are shared across organizations through their parameters only without sharing any raw data. The models are then merged in a way to improve performance compared to individually trained models. Details on the sharing and merging of trained models are provided in \cref{sec:algorithms}.

In order to entice a large group of organizations to participate in a collaborative learning enterprise, we lower barriers-to-entry by minimizing the computational requirements regarding platform, bandwidth, processing power, and storage. Our design choices were guided by this constraint, where we preferred small algorithms, minimal sharing, and streaming processing. Some modifications to the machine learning algorithms designed for traditional batch processing are needed to enable streaming learning. (For a similar approach see\cite{2020river}.)

\section{System Architecture}
\label{sec:architecture}
To achieve the simultaneous goals of supporting real-time analytics of streaming log data and ad hoc offline analytics, we propose a simple directed acyclic graph (DAG) computational model. This model assumes a well-defined data model to facilitate “pluggable” analytics and multi-organizational data federation. The compute flow is based upon a small number of autonomous boxes that are connected in a specific, simple way as illustrated in \cref{fig:a}.
\begin{figure}[htbp]
  \centering
  \label{fig:a}\includegraphics[width=5.4in]{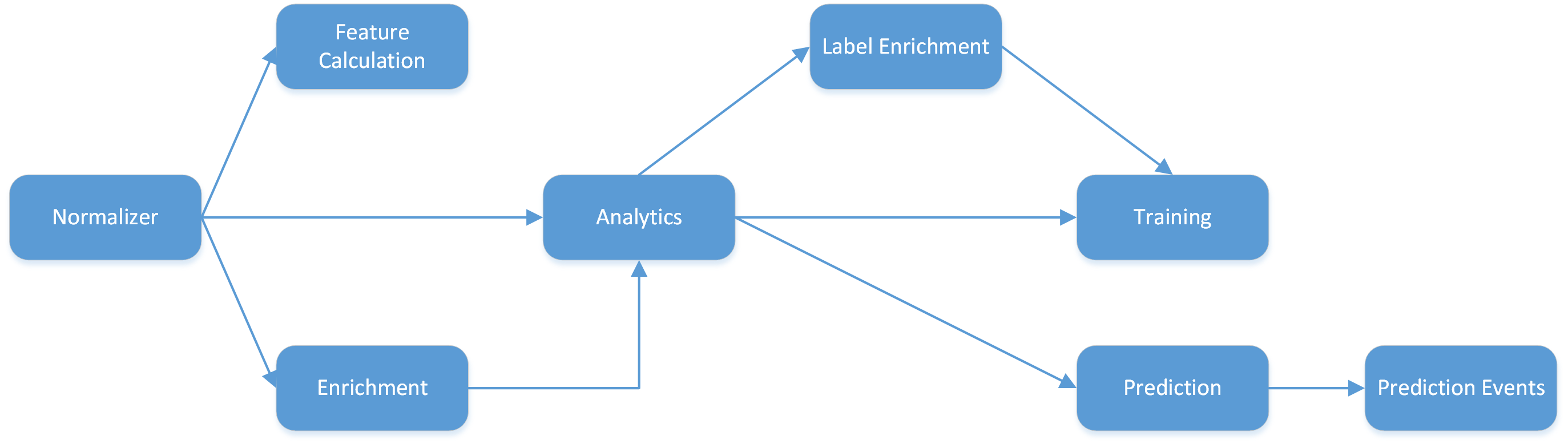}
  \caption{Core computation DAG.}
  \label{fig:computedag}
\end{figure}

The flow of data through the DAG ensures data at each stage is consistent and provides guarantees of how, when and where data is both shared and integrated. Sharing occurs at “stores” where sharable data is integrated across organizations, as shown in \cref{fig:dagwithstores}, in the live stream using well-known model combiners specific to each analytic shared. In this way, the development of a new analytic involves the creation of a local training processor, a model combining routine, and a prediction processor.
\begin{figure}[htbp]
  \centering
  \label{fig:b}\includegraphics[width=6.15in]{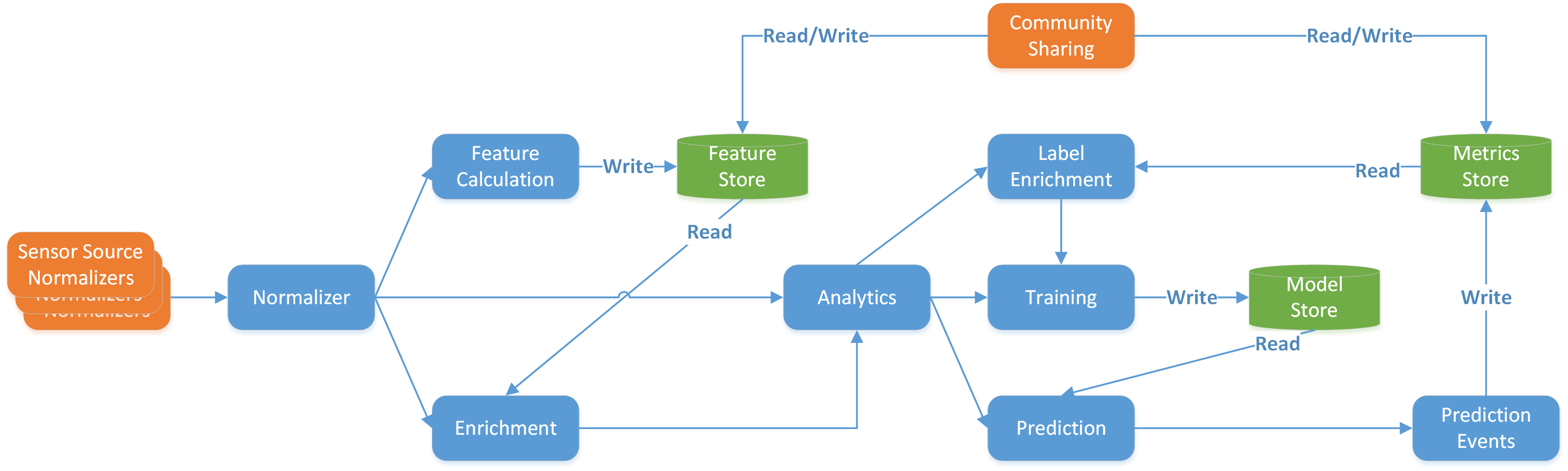}
  \caption{Computation DAG extended with feedback stores where sharing occurs.}
  \label{fig:dagwithstores}
\end{figure}
Extending the sharing through the use of communities will enable data to be aggregated across community members and then re-shared as the “community consensus" for further integration at member sites.

\section{Machine Learning Algorithms}
\label{sec:algorithms}

Finally, any number of prediction algorithms can be plugged into the streaming framework along with methods to merge the separately trained models. We built a prototype for our described system using three example classes of machine learning algorithms, Example \thesection.1 (neural network), Example \thesection.2 (naive Bayes), and Example \thesection.3 (Random Forest), which illustrate how a large number of different classification algorithms might be implemented. The input data for the machine learning algorithms were obtained using existing raw HTTP data logs that were preprocessed to obtain 81-dimensional vectors $\mathbf x$ (features), most of which were processed according to the method described by Oprea, et al.\cite{oprea2018made} Limited malicious labels were derived from \href{www.virustotal.com}{VirusTotal} and \href{www.fireeye.com}{FireEye\texttrademark} where available, and the host traffic ranking was used as a proxy for benign labels. The streaming federated algorithm with scheduled model sharing is described in \cref{alg:ml_algorithm}.

\begin{algorithm}
\caption{Streaming federated algorithm}
\label{alg:ml_algorithm}
\begin{algorithmic}
\WHILE{\textbf{not} past time for scheduled sharing}
\STATE{Receive HTTP record $R$}
\STATE{Preprocess to get input vector $\mathbf{x}$}
\IF{$R$ contains label information}
\STATE{Extract labels $\mathbf{y}$}
\STATE{Perform training iteration with $\mathbf{x}\xrightarrow{}\mathbf{y}$}
\ELSE
\STATE{Perform prediction on $\mathbf{x}$}
\ENDIF
\ENDWHILE
\STATE{Exchange model parameters with collaborators}
\STATE{Merge models}
\STATE{Repeat from top}
\end{algorithmic}
\end{algorithm}

\vspace{0.25in}

\indent \textit{Example \thesection.1} (Neural Network)
To maintain a small model size, for our neural network model, we avoided deep learning and implemented a multilayer perceptron with five hidden layers of size [64, 32, 16, 8, 4] and output size two for the malicious and benign scores. To merge neural network models we perform a weighted average for the weights and biases of each node in the network. So for each weight $w_i$ and bias $b_i$, where $i$ is the model index, we obtain the merged weight $w^\prime$ and bias $b^\prime$ through the weighted averages
\begin{equation}
    w^\prime = \langle \mathbf{a}, \mathbf{w} \rangle \hspace{10pt} \text{and} \hspace{10pt} b^\prime = \langle \mathbf{a}, \mathbf{b} \rangle,
\end{equation}
where $\mathbf{a}$ is the (averaging) weights for the shared models with $||\mathbf{a}||_1 = 1$, and $\mathbf{w}$ and $\mathbf{b}$ are the the weights and biases for all shared models at a given node.

\vspace{5pt}
\indent \textit{Example \thesection.2} (Naive Bayes)
For the Naive Bayes algorithm, we generate two histograms $\mathbf{h}_{k,i}$ for each feature $i$, one histogram for the benign labels and one for the malicious labels, where $k \in \{\text{benign, malicious}\}$. Let $N_k = ||\mathbf{h}_{k,i}||_1$ represent the cumulative numbers of benign and malicious records in the histograms. Then the benign and malicious likelihoods $\mathcal{L}_k$ and evidence $\mathcal{E}$ for each feature $i$ are given by
\begin{equation}
    \mathcal{L}_{k,i} = \frac{\mathbf{h}_{k,i}}{N_k} \hspace{10pt} \text{and}  \hspace{10pt} \mathcal{E}_i = \frac{\sum_k \mathbf{h}_{k,i}}{\sum_k N_k},
\end{equation}
respectively.
The probabilities $p(k \mid \mathbf{x})$ for a log record with feature vector $\mathbf{x} = (x_1, x_2, \ldots, x_n)$ being benign ($k = \text{benign}$) or malicious ($k = \text{malicious}$)  are given by
\begin{equation}
    \label{eq:prob}
        p(k \mid \mathbf{x}) = \frac{N_k}{\sum_k N_k} \prod_{i=1}^n \frac{\mathcal{L}_{k,i}(x_i)}{\mathcal{E}_i(x_i)},
\end{equation}
where $\mathcal{L}_{k,i}(x_i)$ are the likelihoods and $\mathcal{E}_i(x_i)$ is the evidence for the histogram bin corresponding to the value $x_i$ of feature $i$. The proof for \cref{eq:prob} is given in \cref{sec:nb_proof}. Model merging is accomplished by simply summing the histograms from the shared models.

\vspace{5pt}
\indent \textit{Example \thesection.3} (Random Forest)
As an example of an ensemble model, the Random Forest algorithm implements some number $m$ of constituent models that are combined for prediction. The Random Forest algorithm is trained in the normal way. To merge the shared ensembles, we select a sampling of constituent models from each shared ensemble so that we maintain $m$ models in the merged ensemble. 

\section{Summary}
\label{sec:summary}

We presented an automated cyber threat detection system that leverages sharing communities for collaborative, federated learning with a streaming architecture. The computation DAG addresses the need for data privacy and low barriers-to-entry. We completed a prototype with streaming machine learning algorithms with the capability to merge shared models. The results from our internal testing demonstrate the feasibility of our approach and the effectiveness of sharing and merging shared models, however a detailed performance analysis is contingent on more robust labeling.


\section*{Acknowledgments}
We would like to acknowledge the University of Virginia PCORE-CHASE team for their support.

\bibliographystyle{siamplain}
\bibliography{references}

\includepdf[pages=-]{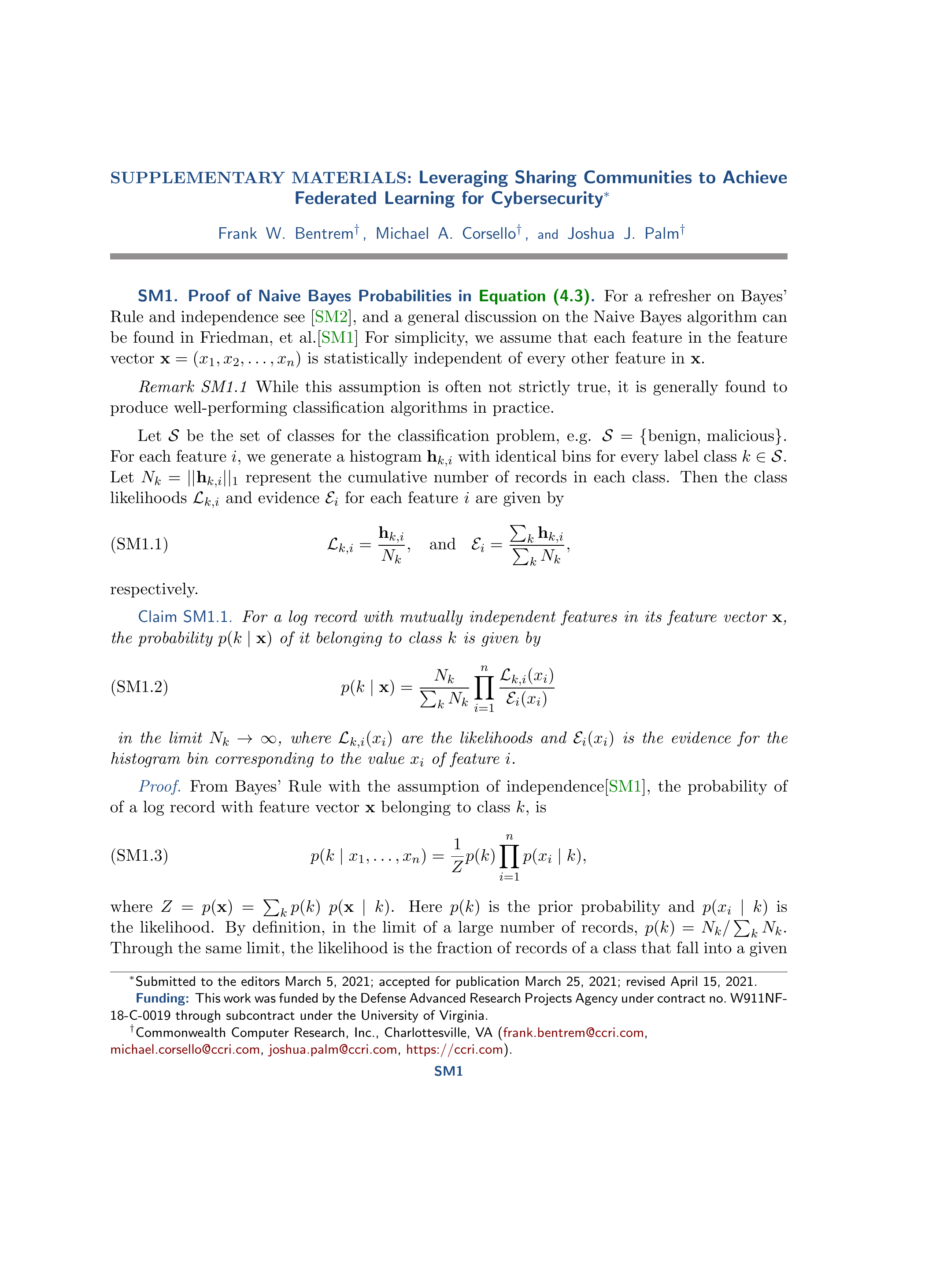}

\end{document}


\maketitle

\section{Proof of Naive Bayes Probabilities in \Cref{eq:prob}}
\label{sec:nb_proof}

For a refresher on Bayes' Rule and independence see \cite{metcalf2016cybersecurity}, and a general discussion on the Naive Bayes algorithm can be found in Friedman, et al.\cite{friedman2001elements} For simplicity, we assume that each feature in the feature vector $\mathbf{x}=(x_1, x_2, \ldots, x_n)$ is statistically independent of every other feature in $\mathbf{x}$. 

\vspace{5pt}
\indent \textit{Remark \thesection.1}
While this assumption is often not strictly true, it is generally found to produce well-performing classification algorithms in practice.
\vspace{5pt}

Let $\mathcal{S}$ be the set of classes for the classification problem, e.g. $\mathcal{S} = \text{\{benign, malicious\}}$. For each feature $i$, we generate a histogram $\mathbf{h}_{k,i}$ with identical bins for every label class $k \in \mathcal{S}$. Let $N_k = ||\mathbf{h}_{k,i}||_1$ represent the cumulative number of records in each class. Then the class likelihoods $\mathcal{L}_{k,i}$ and evidence $\mathcal{E}_i$ for each feature $i$ are given by
\begin{equation}
    \mathcal{L}_{k,i} = \frac{\mathbf{h}_{k,i}}{N_k}, \hspace{10pt} \text{and} \hspace{10pt} \mathcal{E}_i = \frac{\sum_k \mathbf{h}_{k,i}}{\sum_k N_k},
\end{equation}
respectively.

\begin{claim}
For a log record with mutually independent features in its feature vector $\mathbf{x}$, the probability $p(k \mid \mathbf{x})$ of it belonging to class $k$ is given by
\begin{equation} \label{eq:nb_prob}
    p(k \mid \mathbf{x}) = \frac{N_k}{\sum_k N_k} \prod_{i=1}^n \frac{\mathcal{L}_{k,i}(x_i)}{\mathcal{E}_i(x_i)}
\end{equation}\label{clm:prob}
in the limit $N_k \rightarrow \infty$, where $\mathcal{L}_{k,i}(x_i)$ are the likelihoods and $\mathcal{E}_i(x_i)$ is the evidence for the histogram bin corresponding to the value $x_i$ of feature $i$.
\end{claim}

\begin{proof}
From Bayes' Rule with the assumption of independence\cite{friedman2001elements}, the probability of of a log record with feature vector $\mathbf{x}$ belonging to class $k$, is
\begin{equation}
p(k \mid x_{1},\ldots ,x_{n})={\frac {1}{Z}}p(k)\prod _{i=1}^{n}p(x_{i}\mid k),
\end{equation}
where $Z=p(\mathbf {x} )=\sum _{k}p(k)\ p(\mathbf {x} \mid k)$. Here $p(k)$ is the prior probability and $p(x_i\mid k)$ is the likelihood. By definition, in the limit of a large number of records, $p(k) = N_k / \sum_k N_k$. Through the same limit, the likelihood is the fraction of records of a class that fall into a given histogram bin, so that $p(x_i \mid k) = \mathbf{h}_{k,i}(x_i) / N_k = \mathcal{L}_{k,i}(x_i)$, where $\mathbf{h}_{k,i}(x_i)$ is the histogram bin value for the feature value $x_i$. \Cref{eq:nb_prob} is then equivalent to
\begin{equation} \label{eq:nb_prob2}
    p(k \mid \mathbf{x}) = Z^{-1} \frac{N_k}{\sum_k N_k} \prod_{i=1}^n \mathcal{L}_k(x_i).
\end{equation}
Again under the assumption of mutual independence,
\begin{equation} \label{eq:evidence} 
    p(\mathbf{x}) = \prod_{i=1}^n p(x_i), \hspace{10pt} \text{so} \hspace{10pt} Z = \prod_{i=1}^n \frac{\sum_k \mathbf{h}_{k,i}(x_i)}{\sum_k N_k} = \prod_{i=1}^n \mathcal{E}(x_i).
\end{equation}
 Substituting \cref{eq:evidence} into \cref{eq:nb_prob2} results in
\begin{equation}
    p(k \mid \mathbf{x}) = \frac{N_k}{\sum_k N_k} \prod_{i=1}^n \frac{\mathcal{L}_k(x_i)}{\mathcal{E}(x_i)},
\end{equation}
which is the result claimed.
\end{proof}

\bibliographystyle{siamplain}
\bibliography{references}